\documentclass[onecolumn,aps,nofootinbib,superscriptaddress,tightenlines,notitlepage,11pt,pra]{revtex4-1}
\pdfoutput=1
\usepackage[ascii]{inputenc}
\usepackage[bookmarksnumbered,hypertexnames=false,colorlinks=true,linkcolor=blue,urlcolor=blue,citecolor=blue,anchorcolor=green,breaklinks=true,pdfusetitle]{hyperref}
\usepackage[usenames,dvipsnames]{xcolor}
\usepackage{bbm,microtype,amsfonts,amsmath,amssymb,amsthm,graphicx,nicefrac,dsfont,enumitem}
\usepackage[T1]{fontenc}
\usepackage[capitalize]{cleveref}
\usepackage[margin=3cm]{geometry}
\theoremstyle{plain}
\newtheorem{theorem}{Theorem}[section]

\newtheorem{fact}[theorem]{Fact}
\newtheorem{corollary}[theorem]{Corollary}
\newtheorem{proposition}[theorem]{Proposition}

\theoremstyle{definition}

\newtheorem{remark}[theorem]{Remark}
\newtheorem{example}[theorem]{Example}
\DeclareMathOperator{\BL}{BL}
\DeclareMathOperator{\diag}{diag}

\DeclareMathOperator{\tr}{tr}
\DeclareMathOperator{\Her}{Herm}
\DeclareMathOperator{\Pos}{P_\succ}
\DeclareMathOperator{\Posemi}{P_\succeq}
\DeclareMathOperator{\GL}{GL}

\DeclareMathOperator{\TPCP}{TPCP}
\DeclareMathOperator{\TPP}{TPP}
\DeclareMathOperator{\St}{S}
\DeclareMathOperator{\LL}{L}
\newcommand*{\ee}{\mathrm{e}}
\newcommand*{\eps}{\varepsilon}
\newcommand*{\Z}{\mathbb{Z}}
\newcommand*{\X}{\mathbb{X}}
\newcommand*{\Y}{\mathbb{Y}}

\newcommand*{\N}{\mathbb{N}}
\newcommand*{\R}{\mathbb{R}}

\newcommand*{\C}{\mathbb{C}}

\newcommand*{\cE}{\mathcal{E}}
\newcommand*{\cF}{\mathcal{F}}

\newcommand*{\cL}{\mathcal{L}}
\newcommand*{\cI}{\mathcal{I}}

\newcommand*{\cP}{\mathcal{P}}
\newcommand*{\cQ}{\mathcal{Q}}

\newcommand*{\id}{\mathds{1}}
\newcommand*{\ci}{\mathrm{i}} 
\newcommand*{\di}{\mathrm{d}} 
\newcommand{\norm}[1]{\left\lVert#1\right\rVert}
\newcommand{\normU}[1]{{\left\vert\kern-0.25ex\left\vert\kern-0.25ex\left\vert#1\right\vert\kern-0.25ex\right\vert\kern-0.25ex\right\vert}}
\newcommand*{\ket}[1]{| #1 \rangle}
\newcommand*{\bra}[1]{\langle #1 |}

\newcommand{\proj}[1]{|#1\rangle\!\langle #1|}
\allowdisplaybreaks
\begin{document}


\title{Quantum Brascamp-Lieb Dualities}

\author{Mario Berta}
\affiliation{Institute for Quantum Information, RWTH Aachen University, Aachen, Germany}
\author{David Sutter}
\affiliation{IBM Quantum, IBM Research Europe, Z\"urich, Switzerland}
\author{Michael Walter}
\affiliation{Faculty of Computer Science, Ruhr University Bochum, Germany}
\affiliation{Korteweg-de Vries Institute for Mathematics, Institute for Theoretical Physics, Institute for Language, Logic \& Computation, and QuSoft, University of Amsterdam, Netherlands}

\begin{abstract}
Brascamp-Lieb inequalities are entropy inequalities which have a dual formulation as generalized Young inequalities. In this work, we introduce a fully quantum version of this duality, relating quantum relative entropy inequalities to matrix exponential inequalities of Young type. We demonstrate this novel duality by means of examples from quantum information theory\,---\,including entropic uncertainty relations, strong data-processing inequalities, super-additivity inequalities, and many more. As an application we find novel uncertainty relations for Gaussian quantum operations that can be interpreted as quantum duals of the well-known family of `geometric' Brascamp-Lieb inequalities.
\end{abstract}

\maketitle

\section{Introduction}\label{sec:intro}

The classical Brascamp-Lieb (BL) problem asks, given a finite sequence of surjective linear maps~$L_k\colon \R^m \to\R^{m_k}$ and~$q_k\in\R_+$ for $k\in[n]$, for the optimal constant~$C\in\R$ such that~\cite{Brascamp76,Lieb1990,Barthe1998,Bennett2008}
\begin{align}\label{eq:introI}
\int_{\R^m} \prod_{k=1}^n f_k\big(L_k x\big) \,\di x \leq \exp(C)\prod_{k=1}^n\|f_k\|_{1/q_k}
\end{align}
holds for all non-negative functions $f_k\colon \R^{m_k}\to\mathbb{R}_+$, $k\in[n]$, where $\|\cdot\|_p$ denotes the $p$-norm.
Many classical integral inequalities fall into this framework, such as the H\"older inequality, Young's inequality, and the Loomis-Whitney inequality. A celebrated theorem by Lieb asserts that the optimal constant in \cref{eq:introI} can be computed by optimizing over centred Gaussians~$f_k$ alone~\cite{Lieb1990}.

Remarkably, \cref{eq:introI} has a dual, entropic formulation in terms of the \emph{differential entropy}~$H(g):=-\int\,g(x)\log g(x)\, \di x$.
Namely, \cref{eq:introI} holds for all $f_1,\dots,f_n$ as above if, and only if, for all probability densities $g$ on $\R^m$ with finite differential entropy, we have~\cite{Carlen09}
\begin{align}\label{eq:introII}
  H(g)\leq\sum_{k=1}^nq_k H(g_k)+C \, .
\end{align}
Here, $g_k$ denotes the marginal probability density on~$\R^{m_k}$ corresponding to $L_k$, i.e., the push-forward of~$g$ along $L_k$ defined by $\int_{\R^m} \phi(L_k x)g(x) \di x = \int_{\R^{m_k}} \phi(y) g_k(y) \di y$ for all bounded, continuous functions $\phi$ on~$\R^{m_k}$.
The duality between \cref{eq:introI} and \cref{eq:introII} readily generalizes to arbitrary measure spaces and measurable maps~\cite{Carlen09}.

Of particular interest is the so-called \emph{geometric} case where each~$L_k$ is a surjective partial isometry and~$\sum_{k=1}^n q_k \, L_k^\dagger L_k = \id_{\R^m}$~\cite{Ball1989,ball1991shadows,ball_1991_ratios,ball_1997,ball2001convex,Barthe1998}.
In this case, \cref{eq:introI} and \cref{eq:introII} hold with~$C=0$.
This setup includes the H\"older and Loomis-Whitney inequalities.
Equivalently, we are given~$n$ subspaces~$V_k\subseteq\R^m$ (the supports of the $L_k$) such that~$\sum_{k=1}^n q_k \, \Pi_k = \id_{\R^m}$, where $\Pi_k$ denotes the orthogonal projection onto~$V_k$.
In this case we can think of the marginal densities~$g_k$ as functions on $V_k$, namely
\begin{align}\label{eq:generalized marginal}
  g_{V_k}(y) = \int_{V_k^\perp} g(y + z) \di z \qquad \forall y \in V_k \, .
\end{align}
In particular, if~$V_k$ is a coordinate subspace of~$\R^m$ then $g_{V_k}$ is nothing but the usual marginal probability density of the corresponding random variables, justifying our terminology. As a concrete example, let $V_1$, $V_2$ be the two coordinate subspaces of~$\R^2$ and $q_1=q_2=1$; then \cref{eq:introII} amounts to the sub-additivity property of the differential entropy, which is dual to the trivial estimate~$\int_{\R^2} f_1(x_1) f_2(x_2) \,\di x \leq \|f_1\|_1 \|f_2\|_1$.
In contrast, already for three equiangular lines in~$\R^2$ (a `Mercedes star' configuration) and $q_1=q_2=q_3=\frac23$, neither inequality is immediate.

Recently, the BL duality has been extended on the entropic side to not only include entropy inequalities as in \cref{eq:introII} but also relative entropy inequalities in terms of the Kullback-Leibler divergence~\cite{Verdu16}. The dual analytic form then again corresponds to generalized Young inequalities as in \cref{eq:introI} but now for weighted $p$-norms. Interestingly, this extended BL duality covers many fundamental entropic statements from information theory and more. This includes, e.g., hypercontractivity inequalities, strong data processing inequalities, and transportation-cost inequalities~\cite{Verdu17}.

Here, we raise the question how aforementioned BL dualities can be extended in the non-commutative setting.
Our main motivation comes from quantum information theory, where quantum entropy inequalities are pivotal and dual formulations often promise new insights.
BL dualities for non-commutative integration have previously been studied by Carlen and Lieb~\cite{Lieb08}.
Amongst other contributions, they gave BL dualities similar to \cref{eq:introI} - \cref{eq:introII} leading to generalized sub-additivity inequalities for quantum entropy.

In this paper, we extend the classical duality results of~\cite{Verdu16,Verdu17} to the quantum setting\,---\,thereby generalizing Carlen and Lieb's BL duality to the quantum relative entropy and general quantum channel evolutions. In particular, we derive in \cref{sec:BL} a fully quantum BL duality for quantum relative entropy and discuss its properties. In \cref{sec_examples} we then discuss a plethora of examples from quantum information theory that are covered by our quantum BL duality. As novel inequalities, we give quantum versions of the \emph{geometric} Brascamp-Lieb inequalities discussed above, whose entropic form can be interpreted as an uncertainty relation for certain Gaussian quantum operations (\cref{sec:Gaussian}).

\emph{Note added:} Since the first version of our manuscript, our geometric quantum Brascamp-Lieb inequalities from \cref{sec:Gaussian} have been extended to the conditional case \cite{Ligthart20} and to more general Gaussian quantum operations \cite{dePalma21}. We briefly mention these extensions in \cref{sec:Gaussian}.

\medskip\textbf{Notation.}\phantomsection\addcontentsline{toc}{subsection}{Notation}
Let $A$ and $B$ be separable Hilbert spaces.
We denote the set of bounded operators on $A$ by~$\LL(A)$, the set of trace-class operators on $A$ by $\mathrm{T}(A)$, the set of Hermitian operators on $A$ by $\Her(A)$, the set of positive operators on $A$ by $\Pos(A)$, and the set of positive semi-definite operators on $A$ by $\Posemi(A)$. A \emph{density operator} or \emph{quantum state} is a positive semi-definite trace-class operator with unit trace; we denote the set of density operators on $A$ by $\St(A)$. The set of trace-preserving and positive maps from~$\mathrm{T}(A)$ to~$\mathrm{T}(B)$ is denoted by $\TPP(A,B)$ and the set of trace-preserving and \emph{completely} positive maps from~$\mathrm{T}(A)$ to~$\mathrm{T}(B)$ is denoted by~$\TPCP(A,B)$. For $\cE\in\TPP(A,B)$ the \emph{adjoint map}~$\cE^\dagger$, which is a unital and positive map from $\LL(B)$ to $\LL(A)$, is defined by $\tr \cE(X)^\dagger Y = \tr X^\dagger \cE^\dagger(Y)$ for all $X \in \mathrm{T}(A)$ and $Y \in \LL(B)$. When it is clear from the context, we sometimes leave out identity operators, i.e., we may write $\rho_A \sigma_{AB} \rho_B$ for~$(\rho_A \otimes \id_B) \sigma_{AB} (\id_A \otimes \rho_B)$.

The \emph{von Neumann entropy} of a density operator~$\rho\in\St(A)$ is defined as\footnote{The case when~$\rho_A$ does not have full support is covered by the convention $0\log0=0$. Unless specified otherwise, we choose to leave the basis of the logarithm function $\log(\cdot)$ unspecified and write $\exp(\cdot)$ for its inverse function.}
\begin{align*}
  H(\rho) := -\tr \rho\log\rho
\end{align*}
and can be infinite (only) if $A$ is infinite-dimensional.
The \emph{quantum relative entropy} of $\omega\in\St(A)$ with respect to $\tau\in\Posemi(A)$ is given by
\begin{align*}
  D(\omega\|\tau):= \tr \omega(\log \omega-\log \tau)\quad\text{if $\omega\ll \tau$ and as $+\infty$ otherwise}\,,
\end{align*}
where $\omega \ll \tau$ denotes that the support of $\omega$ is contained in the support of $\tau$.
The von Neumann entropy can be expressed as a relative entropy, $H(\rho) = -D(\rho\|\id)$, where $\id$ denotes the identity operator.
For $\rho_{AB} \in \St(A \otimes B)$ with $H(A)_\rho<\infty$, the \emph{conditional entropy of $A$ given $B$} is defined as~\cite{Kuznetsova11}
\begin{align*}
  H(A|B)_{\rho}:= H(A)_\rho - D(\rho_{AB}\| \rho_A \otimes \rho_B),
\end{align*}
where the notation~$H(A)_\rho := H(\rho_A)$ refers to the entropy of the reduced density operator~$\rho_A = \tr_B(\rho)$ on~$A$.
For $A$ and $B$ finite-dimensional we can also write $H(A|B)_\rho = H(AB)_\rho - H(B)_\rho$.

Throughout this manuscript the default is that Hilbert spaces are finite-dimensional unless explicitly stated otherwise (such as in~\cref{sec:Gaussian}).

\section{Brascamp-Lieb duality for quantum relative entropies}\label{sec:BL}
In this section, we describe our main result (\cref{thm_BLrelEnt}) and discuss some of its mathematical properties.

\subsection{Main result}\label{sec:main-sub}
The following result establishes a version of the Brascamp-Lieb dualities of~\cite{Carlen09,Verdu16,Verdu17} for quantum relative entropies.

\begin{theorem}[Quantum Brascamp-Lieb duality]\label{thm_BLrelEnt}
Let $n\in \N$, $\vec{q}=(q_1,\cdots,q_n)\in \R_+^n$, $\vec{\cE}=(\cE_1,\cdots,\cE_n)$ with $\cE_k \in \TPP(A,B_k)$ for $k\in[n]$, $\sigma \in \Pos(A)$, $\vec{\sigma}=(\sigma_1,\cdots,\sigma_n)$ with $\sigma_k \in \Pos(B_k)$ for $k\in[n]$, and $C\in\R$.
Then, the following two statements are equivalent:
\begin{align}
\label{eq_resEntropy}
\sum_{k=1}^n q_k D\big(\cE_k(\rho) \| \sigma_k \big) &\leq D(\rho \| \sigma)+C \quad \forall \rho\in\St(A) \,, \\
\label{eq_resBL}
\tr \exp \! \left(\! \log \sigma \!+\! \sum_{k=1}^n \cE_k^\dagger(\log \omega_k)\! \right)
&\leq \exp(C) \prod_{k=1}^n\! \left\lVert \exp \bigl(\log \omega_k + q_k \log \sigma_k \bigr) \right\rVert_{1/q_k} \; \forall \omega_k\in\Pos(B_k) \,,\!\!
\end{align}
where $\norm{L}_p := (\tr \lvert L\rvert^p)^{\frac{1}{p}}$ is the \emph{Schatten $p$-norm} for~$p\in[1,\infty]$ and an anti-norm for $p\in(0,1]$.\footnote{An \emph{anti-norm} is a non-negative function on~$\Pos(A)$ that is homogeneous ($\norm{\alpha \omega} = \alpha \norm{\omega}$ for $\alpha>0$) and super-additive ($\norm{\omega+\omega'} \geq \norm{\omega} + \norm{\omega'}$) for $\omega$, $\omega'\in\Pos(A)$~\cite{hiai_11}. NB: $\lVert\cdot\rVert_1$ is both a norm and an anti-norm.}
Moreover, \cref{eq_resBL} holds for all $\omega_k\in\Pos(B_k)$ if and only if it holds for all $\omega_k\in\St(B_k)$ with full support.
\end{theorem}

\noindent
We refer to \cref{eq_resEntropy} as a quantum Brascamp-Lieb inequality in \emph{entropic form}, and to \cref{eq_resBL} as a quantum Brascamp-Lieb inequality in \emph{analytic form}.
The latter can be understood as a quantum version of a Young-type inequality.
The two formulations in \cref{eq_resEntropy} and \cref{eq_resBL} encompass a large class of concrete inequalities, as we will see in \cref{sec_examples} below; we are also often interested in identifying the smallest constant~$C\in\R$ such that either inequality holds.
To this end, both directions of \cref{thm_BLrelEnt} are of interest:
\begin{enumerate}
\item To prove quantum entropy inequalities, \cref{thm_BLrelEnt} allows us to alternatively work with matrix exponential inequalities in the analytic form.
That this approach can give crucial insights was already discovered in the original proof of strong sub-additivity of the von Neumann entropy~\cite{LieRus73}, which relied on Lieb's triple matrix inequality for the exponential function (see also~\cite{Frank2013,Sutter2017} for more recent works). We discuss similar examples in \cref{subsec:entropic}.
\item In the commutative setting, we know that for deriving Young-type inequalities it can be beneficial to work in the entropic form~\cite{Carlen04,Carlen09}.
As the quantum relative entropy has natural properties mirroring its classical counterpart, this translates to the non-commutative setting.
We discuss corresponding examples in \cref{subsec:sub-additivity} and \cref{sec:Gaussian}.
\end{enumerate}

\noindent
The proof of \cref{thm_BLrelEnt} relies on the following formula for the Legendre transform of the quantum relative entropy and its dual.

\begin{fact}[Variational formula for quantum relative entropy~\cite{Petz_variational88}]\label{lem_variational}
Let $\sigma \in \Pos(A)$. Then:
\begin{itemize}
\item For all $\rho \in \St(A)$ we have
\begin{align}\label{eq_varRelEnt}
D(\rho \| \sigma) = \sup \limits_{\omega\in\Posemi(A)} \left \lbrace \tr \rho \log\omega - \log \tr \exp(\log\omega + \log \sigma) \right \rbrace \,.
\end{align}
Furthermore, the supremum is attained for $\omega = \exp(\log \rho - \log \sigma)/\tr  \exp(\log \rho - \log \sigma)$.
\item For all $H \in \Her(A)$, we have
\begin{align}\label{eq_legendre}
\log \tr \exp(H +\log \sigma) = \sup_{\omega \in \St(A)} \left \lbrace \tr H \omega - D(\omega \| \sigma) \right \rbrace \, .
\end{align}
Furthermore, the supremum is attained for $\omega= \exp(H +\log \sigma)/\tr \exp(H + \log \sigma)$.
\end{itemize}
\end{fact}

\noindent
These variational formulas are powerful on their own for proving quantum entropy inequalities, as, e.g., the first term in \cref{eq_varRelEnt} only depends on $\rho$ (but not on $\sigma$) and the second term only on $\sigma$ (but not on $\rho$). We refer to~\cite{Sutter2017} for a more detailed discussion.

We mention that Carlen-Lieb use the variational characterization of the von Neumann entropy to derive Brascamp-Lieb dualities and \cite[bottom of page 564]{Lieb08} commented that their proof strategy extends to the relative entropy via Petz's variational expression for the relative entropy (\cref{lem_variational}), which is what is done here.


\begin{proof}[Proof of \cref{thm_BLrelEnt}]
We first show that \cref{eq_resEntropy} implies \cref{eq_resBL}.
Let $H_k:=\log \omega_k$ and define $H\in\Her(A)$ and $\rho\in\St(A)$ by
\begin{align}\label{eq_Hmat}
H:=\sum_{k=1}^n \cE^\dagger_k(H_k)  \qquad \text{and} \qquad \rho := \frac{\exp(H +\log \sigma)}{\tr \exp(H + \log \sigma)} \  \, ,
\end{align}
respectively.
Then,
\begin{align*}
\log \tr \exp \left(  \log \sigma +  \sum_{k=1}^n \cE^\dagger_k(H_k)  \right)
&= \log \tr \exp ( H + \log \sigma ) \\
&= \tr H \rho - D(\rho \| \sigma) \\
&= \sum_{k=1}^n \tr \cE^\dagger_k(H_k) \rho - D(\rho \| \sigma) \\
&\leq C + \sum_{k=1}^n  q_k\Big( \tr \frac{H_k}{q_k} \cE_k(\rho) - D\big(\cE_k(\rho) \| \sigma_k \big) \Big)  \\
&\leq C + \sum_{k=1}^n q_k \log \tr \exp\left(\frac{H_k}{q_k} + \log \sigma_k \right) \, ,
\end{align*}
where we used \cref{eq_legendre} in both the second and the last step and \cref{eq_resEntropy} in the penultimate step.
By substituting $H_k = \log \omega_k$ and taking the exponential on both sides we obtain \cref{eq_resBL}.

We now show that, conversely, \cref{eq_resBL} implies \cref{eq_resEntropy}. Let $\omega=\exp(H)$, with $H$ defined as in \cref{eq_Hmat} in terms of $H_k = \log(\omega_k)$ for $\omega_k \in \Pos_{\sigma_k}(B_k)$ that we will choose later. Then, using \cref{eq_varRelEnt},
\begin{align*}
D(\rho \| \sigma)
&\geq \tr \rho\log \omega - \log \tr \exp(\log\omega + \log \sigma) \\
&= \sum_{k=1}^n  \tr \rho\, \cE^\dagger_k(H_k) - \log \tr \exp\left( \sum_{k=1}^n \cE^\dagger_k(H_k) + \log \sigma \right) \\
&= \sum_{k=1}^n  \tr \cE_k(\rho) \log \omega_k - \log \tr \exp\left( \log \sigma + \sum_{k=1}^n \cE^\dagger_k(\log \omega_k) \right) \\
&\geq \sum_{k=1}^n q_k \left( \tr \cE_k(\rho) \frac{\log \omega_k}{q_k} - \log \tr \exp\Big( \frac{\log \omega_k}{q_k} + \log \sigma_k \Big) \right)-C\\
&= \sum_{k=1}^n q_k D\big(\cE_k(\rho) \| \sigma_k \big)-C\,,
\end{align*}
where the last inequality uses \cref{eq_resBL} and the final step follows from \cref{eq_varRelEnt} provided we choose $\omega_k^{1/q_k}$ as the maximizer for the variational expression of $D\big(\cE_k(\rho) \| \sigma_k \big)$.
\end{proof}

\begin{remark}
As the variational characterizations from \cref{lem_variational} hold in the general $W^*$-algebra setting~\cite{Petz_variational88}, the BL duality in \cref{thm_BLrelEnt} extends to separable Hilbert spaces.
\end{remark}

\begin{remark}\label{rmk:support}
The BL duality in \cref{thm_BLrelEnt} can be extended to $\sigma \in \Posemi(A)$ and $\vec{\sigma}=(\sigma_1,\cdots,\sigma_n)$ with $\sigma_k \in \Posemi(B_k)$ for $k\in[n]$ when
\begin{enumerate}
\item $\cE_k(\rho) \ll \sigma_k$ for all $\rho \in \St(A)$ with $\rho\ll\sigma$
\item $\cE^\dagger(\log\omega_k)\ll\sigma$  for all $\omega_k\in\Posemi(B)$ with $\omega_k\ll\sigma_k$. 
\end{enumerate}
Then, the BL duality still holds but for the alternative conditions
\begin{align*}
\rho\in\St(A)\;\text{with}\;\rho\ll\sigma\;\text{in \cref{eq_resEntropy}}\quad \text{ and } \quad \omega_k\in\Pos(B_k)\;\text{with}\;\omega_k\ll\sigma_k\;\text{in \cref{eq_resBL}.}
\end{align*}
To see this, note that the variational formula in \cref{eq_varRelEnt} still holds for $\sigma\in\Posemi(A)$ as long as $\rho\ll\sigma$ with the supremum taken over $\omega\in\Posemi(A)$ with $\omega\ll\sigma$. Similarly, \cref{eq_legendre} still holds for $H \in \Her(A)$ for $H\ll\sigma$ with the supremum taken over $\omega\in\St(A)$ with $\omega\ll\sigma$. The proof of \cref{thm_BLrelEnt} then also goes through in the more general form.
\end{remark}

In many important applications, we are interested in using \cref{thm_BLrelEnt} either in the situation that $\sigma_k = \mathcal E_k(\sigma)$ for all $k\in[n]$, or in a setting where $\sigma=\id_A$ and $\sigma_k=\id_{B_k}$ for all~$k\in[n]$. In the latter case, \cref{thm_BLrelEnt} specializes to the following equivalence between von Neumann entropy inequalities and Young-type inequalities:

\begin{corollary}\label{cor_BL}
Let $n\in \N$, $\vec{q}=(q_1,\cdots,q_n)\in \R_+^n$, $\vec{\cE}=(\cE_1,\cdots,\cE_n)$ with $\cE_k \in \TPP(A,B_k)$ for $k\in[n]$, and $C\in\R$.
Then, the following two statements are equivalent:
\begin{align}
\label{eq_Entropy}
  H(\rho) &\leq \sum_{k=1}^n q_k H\big(\cE_k(\rho) \big)+C \quad \forall \rho \in \St(A) \, ,\\
\label{eq_BL}
\tr \exp\Big( \sum_{k=1}^n \cE^\dagger_k(\log \omega_k) \Big) &\leq\exp(C)\prod_{k=1}^n \norm{\omega_k}_{1/q_k} \quad \forall \omega_k \in \St(B_k) \,.
\end{align}
\end{corollary}

\noindent
Carlen and Lieb previously proved a variant of \cref{cor_BL} in the $W^*$-algebra setting assuming that the maps $\cE_k^\dagger$ are $W^*$-homomorphisms and that $q_k\in[0,1]$~\cite[Theorem 2.2]{Lieb08}.
One interesting special case is when the~$\cE_k$ are partial trace maps.
The entropic form \cref{eq_Entropy} then corresponds to generalized sub-additivity inequalities for the von Neumann entropy (cf.~\cref{subsec:sub-additivity}).

\subsection{Weighted anti-norms}\label{subsec:antinorm}
In the commutative setting, the right-hand side of \cref{eq_resBL} can conveniently be understood as a product of $\sigma_k$-weighted norms or anti-norms of the operators $\omega_k$~\cite{Verdu16,Verdu17}.
It is natural to ask whether such an interpretation also holds quantumly.
To this end, given~$p \in (0,1]$ and~$\sigma \in \Pos(A)$, define
\begin{align*}
\normU{\omega}_{\sigma,p}
:= \big( \tr \exp( \log \omega^p + \log \sigma) \big)^{\frac{1}{p}}=\norm{\exp \Big(\log \omega + \frac{1}{p} \log \sigma \Big)}_p \, ,
\end{align*}
for all $\omega\in\Pos(A)$.
The following proposition, which follows readily from~\cite{huang19}, shows that~$\normU{\cdot}_{\sigma,p}$ is an anti-norm provided that $p\leq1$.
For $p>1$, it is easy to find $\sigma\in\Pos(A)$ such that the functional~$\normU{\cdot}_{\sigma,p}$ is neither a norm nor an anti-norm.

\begin{proposition}\label{prp:antinorm}
For $p \in (0,1]$ and $\sigma \in \Pos(A)$, $\normU{\cdot}_{\sigma,p}$ is homogeneous and concave, hence an anti-norm.
\end{proposition}
\begin{proof}
Clearly, $\normU{\cdot}_{\sigma,p}$ is homogeneous.
Since moreover $p\in(0,1]$, \cite[Lemma~D.1]{huang19} asserts that its concavity on the set of positive matrices is equivalent to the concavity of its~$p$-th power, i.e.,
\begin{align}\label{eq_traceLieb}
\omega \mapsto \tr  \exp( p \log \omega + H ),
\end{align}
where $H = \log\sigma$.
A well-known result of Lieb~\cite{Lieb73} states that \cref{eq_traceLieb} is indeed concave for any Hermitian matrix~$H$.
Thus, $\normU{\cdot}_{\sigma,p}$ is concave.
As a consequence of homogeneity and concavity, we obtain that $\normU{\cdot}_{\sigma,p}$ is super-additive, as~$\normU{\omega+\omega'}_{\sigma,p} = 2 \normU{\frac{1}{2}\omega+ \frac{1}{2}\omega'}_{\sigma,p} \geq \normU{\omega}_{\sigma,p} + \normU{\omega'}_{\sigma,p}$ for all $\omega$, $\omega'\in\Pos(A)$.
We conclude that $\normU{\cdot}_{\sigma,p}$ is an anti-norm.
\end{proof}

\noindent
Thus, the quantum Brascamp-Lieb inequality in its analytic form \cref{eq_resBL} can be written as
\begin{align}
\tr \exp \left( \log \sigma + \sum_{k=1}^n \cE_k^\dagger(\log \omega_k) \right)
\leq\exp(C)\prod_{k=1}^n \normU{\omega_k}_{\sigma_k,1/q_k}\quad\forall\omega_k\in\St(B_k)\,,
\end{align}
where, assuming that all $q_k \geq 1$, the right-hand side can be interpreted in terms of anti-norms, pleasantly generalizing \cref{eq_BL}.

\subsection{Convexity and tensorization}\label{sec:additivity}
For fixed $n\in \N$, $\vec{\cE}=(\cE_1,\cdots,\cE_n)$ with $\cE_k \in \TPP(A,B)$, $\sigma \in \Pos(A)$, and $\vec{\sigma}=(\sigma_1,\cdots,\sigma_n)$ with $\sigma_k \in \Pos(B_k)$, we define the \emph{Brascamp-Lieb (BL) set} as
\begin{align*}
\BL\Big(\vec{\cE},\vec{\sigma},\sigma\Big) := \Big\{\big(\vec{q},C\big) \in \R_+^n\times\R :~\text{\cref{eq_resEntropy}/\cref{eq_resBL} holds} \Big\}\,.
\end{align*}
We record the following elementary property.

\begin{proposition}[Convexity]
The set $\BL(\vec{\cE},\vec{\sigma},\sigma)$ is convex.
\end{proposition}
\begin{proof}
We use the characterization using the entropic form \cref{eq_resEntropy}.
Let~$(\vec{q}^{(i)},C^{(i)})\in\BL(\vec{\cE},\vec{\sigma},\sigma)$ for $i \in \{1,2\}$.
Let~$\theta\in[0,1]$ and $(\vec q, C)$ the corresponding convex combination, i.e., $\vec{q} := \theta \, \vec{q}^{(1)} + (1-\theta) \, \vec{q}^{(2)}$ and $C := \theta \, C^{(1)} + (1-\theta) \, C^{(2)}$.
Then, for all $\rho\in \St(A)$,
\begin{align*}
  \sum_{k=1}^n q_k D\big(\cE_k(\rho) \| \sigma_k \big)
&= \theta \sum_{k=1}^n q^{(1)}_k D\big(\cE_k(\rho) \| \sigma_k \big)
+ (1 - \theta) \sum_{k=1}^n q^{(2)}_k D\big(\cE_k(\rho) \| \sigma_k \big) \\
&\leq \theta \Bigl( D(\rho \| \sigma) + C^{(1)} \Bigr)
+ (1 - \theta) \Bigl( D(\rho \| \sigma) + C^{(2)} \Bigr)
= D(\rho \| \sigma) + C\,.
\end{align*}
Thus, $(\vec{q},C)\in\BL(\vec{\cE},\vec{\sigma},\sigma)$.
\end{proof}

In the commutative case, the BL set satisfies a \emph{tensorization property}~\cite[Section V.B]{Verdu17}, and we can ask if a similar property holds in the non-commutative case as well.
Namely, do we have that for $\left(\vec{q},C^{(i)}\right)\in \BL\Big(\vec{\cE}^{(i)},\vec{\sigma}^{(i)},\sigma^{(i)}\Big)$ with $i\in\{1,2\}$ and
\begin{align*}
\vec{\cE}:=\left(\cE_1^{(1)}\otimes\cE_1^{(2)},\dots,\cE_n^{(1)}\otimes\cE_n^{(2)}\right)\quad\text{as well as}\quad\vec{\sigma}:=\left(\sigma_1^{(1)}\otimes\sigma_1^{(2)},\dots,\sigma_n^{(1)}\otimes\sigma_n^{(2)}\right)
\end{align*}
that
\begin{align}\label{eq:tensor}
\left(\vec{q},C^{(1)}+C^{(2)}\right)\stackrel{?}{\in}\BL\left(\vec{\cE},\vec{\sigma},\sigma^{(1)}\otimes\sigma^{(2)}\right)\,.
\end{align}
As we will see in several examples (\cref{sec_examples}), tensorization does in general not hold due to the potential presence of entanglement.
Indeed, the problem of deciding in which case \cref{eq:tensor} holds can be understood as a general information-theoretic \emph{additivity problem}, which contains the (non-)additivity for the minimum output entropy as a special case (cf. \cref{eq:minoutadd} in \cref{subsec:minoutent}).

\section{Applications of quantum Brascamp-Lieb duality}\label{sec_examples}
The purpose of this section is to present examples from quantum information theory where the duality from \cref{thm_BLrelEnt} is applicable. The majority of examples concern entropy inequalities that are of interest from an operational viewpoint. \Cref{thm_BLrelEnt} then shows that all entropy inequalities of suitable structure have a dual formulation as an analytic inequality, and vice versa. Depending on the scenario, one form may be easier to prove than the other, and we find that these reformulations often give additional insight.

\subsection{Generalized (strong) sub-additivity}\label{subsec:sub-additivity}
In this section, we discuss entropy inequalities that generalize the sub-additivity and strong sub-additivity properties of the von Neumann entropy. Recall that the latter states that~$H(AB) + H(BC) \geq H(ABC) + H(B)$ for $\rho_{ABC}\in\St(A \otimes B \otimes C)$~\cite{LieRus73}.

We first state the following result from~\cite[Theorem 1.4 \& Theorem 3.1]{Lieb08}, which gives generalized sub-additivity relations and their dual analytic form. Here, the second argument in the relative entropy is always equal to the identity. Throughout this section, all quantum channels are given by partial trace channels.

\begin{corollary}[Quantum Shearer and Loomis-Whitney inequalities,~\cite{Lieb08}]\label{cor:lieb-carlen}
Let~$S_1$, \ldots, $S_n$ be non-empty subsets of $[m]$ such that every~$s\in[m]$ belongs to at least~$p$ of those subsets.
Then, the following inequalities hold and are equivalent:
\begin{align}
\label{eq_Shearer_Ent}
H(A_1 \dots A_m) &\leq \frac{1}{p} \sum_{k=1}^n H(\{A_s\}_{s\in S_k}) \quad \forall \rho \in \St(A_1 \otimes \ldots \otimes A_m) \, ,  \\
\label{eq_Shearer_Func}
\tr \exp \left( \sum_{k=1}^n \id_{\bar S_k} \otimes \log \omega_{S_k}  \right) & \leq \prod_{k=1}^n \norm{\omega_{S_k}}_p \quad \forall \omega_{S_k} \in \St(\otimes_{s\in S_k} A_s) \, ,
\end{align}
where $\bar S$ denotes the complement of a subset~$S$ of~$[m]$.
\end{corollary}

\noindent
Inequalities in the form of \cref{eq_Shearer_Ent} have been termed \emph{quantum Shearer's inequalities} and their analytic counterparts as in \cref{eq_Shearer_Func} are known as \emph{quantum Loomis-Whitney inequalities}.
Interestingly, and as explained in~\cite[Section 1.3]{Lieb08}, the latter cannot directly be deduced from standard matrix trace inequalities such as Golden-Thompson combined with Cauchy-Schwarz.
That \cref{eq_Shearer_Ent} and \cref{eq_Shearer_Func} are equivalent follows from \cref{cor_BL} by choosing $C=0$, $q_k = \frac{1}{p}$, and $\cE_k(\cdot)=\tr_{\bar S_k}(\cdot)$. The following result provides a conditional version of the quantum Shearer inequality with side information.

\begin{proposition}[Conditional quantum Shearer inequality]\label{cor:lieb-carlen conditional}
Let~$S_1$, \ldots, $S_n$ be non-empty subsets of $[m]$ such that every~$s\in[m]$ belongs to \emph{exactly}~$p$ of those subsets.
Then,
\begin{align}\label{eq_Shearer_Ent_cond}
H(A_1 \dots A_m|B) &\leq \frac{1}{p} \sum_{k=1}^n H(\{A_s\}_{s\in S_k}|B) \quad \forall \rho \in \St(A_1 \otimes \ldots \otimes A_m \otimes B) \, .
\end{align}
\end{proposition}

\noindent
For $n=2$, $S_1 = \{1\}$, $S_2 = \{2\}$, $p=1$, \cref{eq_Shearer_Ent_cond} reduces to $H(A_1 A_2 | B) \leq H(A_1 | B) + H(A_2 | B)$, which is equivalent to the strong sub-additivity of von Neumann entropy.\footnote{Our quantum BL duality (\cref{thm_BLrelEnt}) does not directly provide a dual analytic form for the strong sub-additivity of von Neumann entropy or more generally \cref{eq_Shearer_Ent_cond}. Rather, in \cref{sec:data-processing} we provide a dual analytic form for the (a priori more general) \emph{data processing inequality} of the quantum relative entropy.}

Note that, in contrast to \cref{cor:lieb-carlen}, in the conditional case it is not enough to assume that every $s\in[m]$ belongs to at least~$p$ of the subsets.
This is clear from the following proof.
For a concrete counterexample, note that for $n=2$, $S_1=S_2=\{1\}$, $S_3=\{2\}$, $p=1$, \cref{eq_Shearer_Ent_cond} is violated for, e.g., a maximally entangled state between $A_1$ and $B$.

\begin{proof}[Proof of \cref{cor:lieb-carlen,cor:lieb-carlen conditional}]
We adapt the argument of~\cite{Lieb08} to the conditional case.
If $S$ and $T$ are two subsets of $[m]$ then strong sub-additivity implies that
\begin{align*}
  H(\{A_s\}_{s\in S \cup T}|B) + H(\{A_s\}_{s\in S \cap T}|B)
\leq H(\{A_s\}_{s\in S}|B) + H(\{A_s\}_{s\in T}|B) \, .
\end{align*}
This means that we obtain a stronger version of \cref{eq_Shearer_Ent_cond} if we replace any two subsets $S_k$, $S_l$ by $S_k \cup S_l$, $S_k \cap S_l$.
Moreover, each such replacement preserves the number of times that any~$s\in[m]$ is contained in the subsets $S_1,\dots,S_n$.
We can successively apply replacement steps until we arrive at the situation where $S_k \subseteq S_l$ or $S_l \subseteq S_k$ for any two subsets.
Without loss of generality, this means that it suffices to prove \cref{cor:lieb-carlen,cor:lieb-carlen conditional} in the case that $S_1 \supseteq \dots \supseteq S_n$.
In this case, $S_1 = \dots = S_p = [m]$, since each $s\in[m]$ is contained in at least~$p$ of the subsets.
The corresponding inequality \cref{eq_Shearer_Ent_cond} can thus be simplified to
\begin{align*}
  0 \leq \sum_{k=p+1}^n H(\{A_s\}_{s\in S_k}|B).
\end{align*}
If $B=\emptyset$, as in \cref{cor:lieb-carlen}, this inequality holds since the von Neumann entropy is never negative.
And if each $s\in[m]$ belongs to \emph{exactly}~$p$ of the subsets, as in \cref{cor:lieb-carlen conditional}, then $S_{p+1} = \dots = S_n = \emptyset$, so the inequality holds trivially.
\end{proof}

\begin{remark}
\cref{cor:lieb-carlen} and \cref{cor:lieb-carlen conditional} also hold for separable Hilbert spaces, as the variational characterizations from \cref{lem_variational} hold in the general $W^*$-algebra setting~\cite{Petz_variational88}.
\end{remark}

\subsection{Brascamp-Lieb inequalities for Gaussian quantum operations}\label{sec:Gaussian}

In this section, we present quantum versions of the classical Brascamp-Lieb inequalities as in \cref{eq:introI} and \cref{eq:introII}, where probability distributions on~$\R^m$ are replaced by quantum states on~$\LL^2(\R^m)$, the Hilbert space of square-integrable wave functions on~$\R^m$.
We focus on the \emph{geometric} case discussed in the introduction.
The marginal distribution with respect to a subspace~$X\subseteq\R^m$ has the following natural quantum counterpart.
Define a TPCP map~$\cE_X$ as the composition of the unitary $\LL^2(\R^m) \cong \LL^2(X) \otimes \LL^2(X^\perp)$ with the partial trace over the second tensor factor.
Given a density operator~$\rho$ on $\LL^2(\R^m)$, we can think of \[ \rho_X=\cE_X(\rho) \] as the \emph{reduced density operator corresponding to~$X$}.
This is the natural quantum version of the marginal probability density in \cref{eq:generalized marginal} of the introduction.
Indeed, if we identify $\rho$ with its kernel in~$\LL^2(\R^m \times \R^m)$, and likewise for $\rho_k$, then we have the completely analogous formula
\begin{align*}
  \rho_k(y,y') = \int_{X^\perp} \rho\bigl(y + z, y' + z\bigr) \, \di z
  \qquad \forall y, y' \in X \, .
\end{align*}
This definition is very similar in spirit to the quantum addition operation in the quantum entropy power inequality of~\cite{Koenig14} (see also \cite{Koenig16,dePalma18}) and in fact contains the latter as a special case. In linear optical terms, $\rho_X$ can be interpreted as the reduced state of $\dim X$~many output modes obtained by subjecting~$\rho$ to a network of beamsplitters with arbitrary transmissivities.

The following result establishes quantum versions of the Brascamp-Lieb dualities as in \cref{eq:introI} and \cref{eq:introII} for the geometric case.

\begin{proposition}[Geometric quantum Brascamp-Lieb inequalities]\label{prop:geom q}
Let $X_1,\dots,X_n \subseteq \R^m$ be subspaces and let $q_1,\dots,q_n\geq0$ such that $\sum_{k=1}^n q_k \Pi_k = \id_{\R^m}$, where $\Pi_k$ denotes the orthogonal projection onto~$X_k$.
Then, for all $\rho \in \St(\LL^2(\R^m))$ with finite second moments,
\begin{align}\label{eq:geom BL entropic}
H(\rho) &\leq \sum_{k=1}^n q_k H(\rho_{X_k}) \,.
\end{align}
Furthermore, for all $\omega_{X_k} \in \St(\LL^2(X_k))$ such that $\exp\Bigl( \sum_{k=1}^n \id_{X_k^\perp} \otimes \log \omega_{X_k} \Bigr)$ has finite second moments, it holds that%
\footnote{As we only prove \cref{eq:geom BL entropic} for states with finite second moments, we can not apply \cref{thm_BLrelEnt} directly to obtain \cref{eq:analytic BL entropic} (see end of proof).
Removing this assumption is an interesting open problem. }
\begin{align}\label{eq:analytic BL entropic}
\tr \exp\Big( \sum_{k=1}^n \id_{X_k^\perp} \otimes \log \omega_{X_k} \Big) &\leq \prod_{k=1}^n \norm{\omega_{X_k}}_{1/q_k}\,.
\end{align}
\end{proposition}

\noindent
Note that if~$X_k$ is spanned by a subset $S_k\subseteq[m]$ of the~$m$~coordinates of~$\R^m$, then $\rho_{X_k}$ is nothing but the reduced density matrix of subsystems~$S_k$, which appears on the right-hand side of the quantum Shearer inequality \cref{eq_Shearer_Ent}. Thus, \cref{prop:geom q} implies \cref{cor:lieb-carlen} in the case that all $s\in[m]$ are contained in \emph{exactly} $p$ of the subsets~$S_k$.

To establish \cref{prop:geom q}, we will first prove the entropic form \cref{eq:geom BL entropic} using a quantum version of the heat flow approach from~\cite{Carlen04,Barthe2004} (cf.~the recent works~\cite{carlen2014analog,carlen2017gradient,carlen2018noncommutative} on entropy inequalities for quantum Markov semigroups).
We assume some familiarity with Gaussian quantum systems (see, e.g.,~\cite{holevo_book}) and follow the framework of K\"onig and Smith~\cite{Koenig14}, which holds under regularity assumptions on the quantum state, which were subsequently removed by De Palma and Trevisan~\cite{dePalma18}.

Let $X \subseteq \R^m$ be a subspace and~$m_X$ its dimension. For all~$x\in X$, define position and momentum operators on~$\LL^2(X)$ by $(Q_{X,x} \psi)(y) := (x \cdot y) \psi(y)$ and~$P_{X,x} := -\ci\partial_x$.
Denote by $\mathcal N_t^X$ the \emph{non-commutative heat flow} or \emph{heat semigroup}~\cite{Koenig14,dePalma18}, which is a one-parameter semi-group, meaning $\mathcal N_0^X = \id$ and $\mathcal N_s^X \circ \mathcal N_t^X = \mathcal N_{s+t}^X$ for $s,t\geq0$. On a suitable domain it is generated by
\begin{align*}
\cL_X := - \frac 1 4 \sum_{j=1}^{m_X} [Q_{X,e_j}, [Q_{X,e_j}, \rho]] + [P_{X,e_j}, [P_{X,e_j}, \rho]]\,,
\end{align*}
where $\{e_j\}_{j=1}^{m_X}$ is an arbitrary orthonormal basis of $X$ (but we will not directly use this specific form). For every~$t\geq0$, $\mathcal N^X_t$ is a Gaussian TPCP map, hence fully determined by its action on covariance matrices and mean vectors,%
\footnote{The \emph{mean vector} of a quantum state~$\rho$ on~$\LL^2(X)$ is the linear form~$\mu$ on $X\oplus X$ given by~$\mu(v) = \tr \rho R_{X,v}$, where $R_{X,v} = Q_{X,x} + P_{X,y}$ for $v=(x,y)$; the \emph{covariance matrix} of $\rho$ is the quadratic form $\Sigma$ defined by $\Sigma(v,v') = \tr \rho \{R_{X,v} - \mu(v), R_{X,v'} - \mu(v')\}$.}
which is given by
\begin{align}\label{eq:heat flow}
  \Sigma \mapsto \Sigma + t \id \quad\text{and}\quad \mu \mapsto \mu \, .
\end{align}
In particular, the heat flow is independent of the choice of orthonormal basis in~$X$. The generalized partial trace maps~$\cE_X\colon\rho\mapsto\rho_X$ defined above are also Gaussian and act by
\begin{align}\label{eq:gaussian marginal}
  \Sigma \mapsto \Sigma|_X \quad\text{and}\quad \mu \mapsto \mu|_X \, ,
\end{align}
where $\mu|_X$ denotes the restriction of $\mu$ onto $X\oplus X$ and likewise for $\Sigma|_X$.
The non-commutative heat flow is compatible with the maps $\cE_X$, i.e.,
\begin{align*}
  \cE_X \circ \mathcal N_t^{\R^m} = \mathcal N_t^X \circ \cE_X.
\end{align*}
Indeed, since both channels (and hence their composition) are Gaussian, it suffices to verify that the action commutes on the level of mean vectors and covariance matrices, and the latter is clear from \cref{eq:heat flow} and \cref{eq:gaussian marginal}. See also \cite[Lemma 2]{dePalma18}. Thus, we may unambiguously introduce the notation
\begin{align}\label{eq:evolved gaussian marginal}
  \rho^{(t)}_X := \cE_X\big(\mathcal N_t^{\R^m}(\rho)\big) = \mathcal N_t^X\big(\cE_X(\rho)\big) = \mathcal N_t^X\big(\rho_X\big)
\end{align}
for the reduced density operator on~$\LL^2(X)$ at time~$t$. Similarly, we may show that~$\cE_X$ is compatible with phase-space translations (cf.~\cite[Lemma XI.1]{Koenig14}).
For $x\in X$, define the unitary one-parameter groups
\begin{align*}
\text{$\cQ_{X,x}^{(t)}(\rho) := \ee^{-\ci t P_{X,x}} \rho \, \ee^{\ci t P_{X,x}}$ and~$\cP_{X,x}^{(t)}(\rho) := \ee^{\ci t Q_{X,x}} \rho \, \ee^{-\ci t Q_{X,x}}$.}
\end{align*}
They are Gaussian, leave the covariance matrices invariant, and send mean vectors~$\mu \mapsto \mu + t (x^T,0)$ and~$\mu + t (0,x^T)$, respectively.
By comparing with \cref{eq:gaussian marginal}, we find that
\begin{align}\label{eq:gaussian marginal vs translations}
  \cQ_{X,x}^{(t)} \circ \cE_X = \cE_X \circ \cQ_{\R^m,x}^{(t)}
  \quad\text{and}\quad
  \cP_{X,x}^{(t)} \circ \cE_X = \cE_X \circ \cP_{\R^m,x}^{(t)} \, .
\end{align}

In the following we shall make use of two crucial properties of the heat flow that will allow us to `linearize' the proof of the entropy inequality:
First, the entropy of~$\rho_X^{(t)}$ grows logarithmically as $t \rightarrow \infty$ and becomes \emph{asymptotically independent} of the state $\rho$, as proved in~\cite[Corollary~III.4]{Koenig14} and~\cite[Theorem 5]{dePalma18}:
\begin{align*}
  \big|H\big(\rho^{(t)}_X\big) / m_X - (1 - \log 2 + \log t) \big| \rightarrow 0\,.
\end{align*}
In particular, this implies that any valid inequality of the form \cref{eq:geom BL entropic} must satisfy the inequality
\begin{align}\label{eq:asymptotic}
  \sum_{k=1}^n q_k m_{X_k} \geq m \, ,
\end{align}
since this is precisely equivalent to the validity of \cref{eq:geom BL entropic} as $t\to\infty$.
For us, this condition follows by taking the trace on both sides of our assumption that~$\sum_{k=1}^n q_k \Pi_k = \id_{\R^m}$. To state the second property of the heat flow that we will need, we momentarily assume sufficient regularity of the states under consideration, following~\cite{Koenig14}. Then, the \emph{Fisher information} of a one-parameter family of states~$\big\{\sigma^{(s)}\big\}$ is defined as
\begin{align*}
J\big(\big\{\sigma^{(s)}\big\}\big) := \partial^2_{s=0} D\big(\sigma^0\big\|\sigma^{(s)}\big)\,.
\end{align*}
It satisfies the following version of the data processing inequality~\cite[Theorem IV.4]{Koenig14}:
For any TPCP map $\mathcal E$,
\begin{align}\label{eq:fisher dpi}
  J\big(\big\{\mathcal E\big(\sigma^{(s)}\big)\big\}\big) \leq J\big(\big\{\sigma^{(s)}\big\}\big) \, .
\end{align}
For a covariant family of the form~$\sigma^{(s,K)} := \ee^{\ci s K} \sigma \, \ee^{-\ci s K}$, the Fisher information can be computed as~\cite[Lemma IV.5]{Koenig14}
\begin{align}\label{eq:fisher covariant}
  J\big(\{\sigma^{(s,K)}\}\big) = \tr \sigma [K,[K,\log \sigma]]\,.
\end{align}
We can now state the \emph{quantum de Bruijn identity}~\cite[Theorem V.1]{Koenig14}, which computes the derivative of the entropy along the heat flow in terms of the Fisher information:
\begin{align}\label{eq:de bruijn}
  \partial_t H\big(\rho^{(t)}_X\big) = \frac 1 4 J\big(\rho^{(t)}_X\big) \, ,
\end{align}
where the \emph{total Fisher information} $J(\sigma_X)$ of a state $\sigma_X$ on $\LL^2(X)$ is defined by
\begin{align}\label{eq:total fisher}
  J(\sigma_X) := \sum_{j=1}^{m_X} J\big(\big\{\sigma^{(s,Q_{X,e_j})}\big\}\big) + J\big(\big\{\sigma^{(s,P_{X,e_j})}\big\}\big)\, ,
\end{align}
for an arbitrary orthonormal basis~$\{e_j\}_{j=1}^{m_X}$ of~$X$. While above we assumed regularity, the Fisher information $J(\sigma_X)$ can be defined for any state with finite second moments, and the de Bruijn identity~\eqref{eq:de bruijn} generalizes as well~\cite[Definition~7 \& Proposition~1]{dePalma18}.%
\footnote{\label{foot:depalma}%
The key idea is to first define an \emph{integral} version of the Fisher information~\cite[Definition~6]{dePalma18}.
In the setting without side information, this is defined for a state~$\sigma_X$ on~$\LL^2(X)$ and for~$t>0$ by~$\Delta(\sigma_X)(t) := I(X:V)_{\sigma_{XV}(t)}$, where~$\sigma_{XV}(t)$ denotes the classical-quantum state with~$V$ is a multivariate Gaussian random variable with covariance matrix~$t (I_X \oplus I_X)$ and~$\sigma_{X|V=v} = \mathcal D_{X,v} \sigma_X \mathcal D_{X,v}^\dagger$, with~$\mathcal D_{X,v} = \mathcal Q_{X,x}^{(1)} \circ \mathcal P_{X,y}^{(1)}$ for $v=(x,y) \in X \oplus X$; one also sets~$\Delta_X(\sigma_X)(0) := 0$.
This is well-defined for any state~$\sigma_X$ and satisfies a finitary version of the de Bruijn identity~\cite[Theorem~1]{dePalma18}.
Moreover, if~$\sigma_X$ is a state with finite energy then~$\Delta(\sigma_X)(t)$ is continuous, increasing, and concave as a function of~$t\geq0$.
Hence, for such states one can define the Fisher information~$J(\sigma_X)$ as the (right) derivative of~$\Delta(\sigma_X)(t)$ at $t=0$, that is, as $J(\sigma_X) = \lim_{t\downarrow 0} \frac {\Delta(\sigma_X)(t)} t$~\cite[Definition~7]{dePalma18}.
Then the de Bruijn identity~\eqref{eq:de bruijn} follows directly from its finitary version~\cite[Proposition~1]{dePalma18}.}

\begin{proof}[Proof of \cref{prop:geom q}]
We first prove the entropic \cref{eq:geom BL entropic} by considering~$\rho^{(t)} := \rho^{(t)}_{\R^m}$ for $t\geq0$.
As $t\to\infty$, \cref{eq:geom BL entropic} holds up to arbitrarily small error, as explained below \cref{eq:asymptotic}.
To show its validity at~$t=0$, we would therefore like to argue that~$\partial_t H(\rho^{(t)}) \geq \sum_{k=1}^n q_k \, \partial_t H\big(\rho^{(t)}_{X_k}\big)$ for all~$t\geq0$.
In view of the de Bruijn identity in \cref{eq:de bruijn} and \cref{eq:evolved gaussian marginal}, it suffices to establish the following \emph{super}-addivity property of the Fisher information for all states~$\sigma$ on $\LL^2(\R^m)$ with finite second moment:
\begin{align}\label{eq:fisher ieq}
  \sum_{k=1}^n q_k J(\sigma_{X_k}) \leq J(\sigma) \, .
\end{align}
We first prove this under the regularity assumptions of \cite{Koenig14}, so that \cref{eq:fisher covariant} applies. We will abbreviate~$Q_j := Q_{\R^m,e_j}$ and~$P_j := P_{\R^m,e_j}$, where~$\{e_j\}_{j=1}^m$ is the standard basis of~$\R^m$.
For all~$x \in X_k$, it holds that
\begin{align*}
  J\big(\big\{\sigma_{X_k}^{(s,Q_{X_k,x})}\big\}\big)
&= J\big(\big\{\cP^{(s)}_{X,x}(\cE_{X_k}(\sigma))\big\}\big) \\
&= J\big(\big\{\cE_{X_k}(\cP^{(s)}_{\R^m,x}(\sigma))\big\}\big) \\
&\leq J\big(\big\{\cP^{(s)}_{\R^m,x}(\sigma)\big\}\big) \\
&= \tr \sigma [Q_{\R^m,x}, [Q_{\R^m,x}, \log \sigma]] \\
&= \sum_{j,j'=1}^m x_j x_{j'} \tr \sigma [Q_j, [Q_{j'}, \log\sigma]] \\
&= \sum_{j,j'=1}^m \big(x x^T\big)_{j,j'} \tr \sigma [Q_j, [Q_{j'}, \log\sigma]] \,
\end{align*}
where the second step is by the compatibility of phase-space translations and generalized partial trace~\eqref{eq:gaussian marginal vs translations}, the third step uses the data-processing inequality for the Fisher information \cref{eq:fisher dpi}, and the fourth step follows from \cref{eq:fisher covariant}. If we apply the same argument to $J\big(\big\{\sigma_{X_k}^{(s,P_{X,x})}\big\}\big)$ and sum both inequalities over an orthonormal basis $\{x\}$ of $X_k$, we obtain
\begin{align}\label{eq:step1}
  J(\sigma_{X_k})
\leq 
J(\sigma, \Pi_k) \, ,
\end{align}
where we used the shortcut $J(\sigma, A) := \sum_{j,j'=1}^m A_{j,j'} \bigl( \tr \sigma [Q_j, [Q_{j'}, \log\sigma]] + \tr \sigma [P_j, [P_{j'}, \log\sigma]] \bigr)$ for any positive semidefinite $m\times m$ matrix~$A$, which is linear in $A$. Thus, our assumption that~$\sum_{k=1}^n q_k \Pi_k = \id_{\R^m}$ (with all~$q_k\geq0$) implies the desired inequality:
\begin{align}\label{eq:step2}
  \sum_k q_k \, J(\sigma_{X_k}) \leq \sum_k q_k \, J(\sigma, \Pi_{k}) = J(\sigma, \sum_k q_k \Pi_{k}) = J(\sigma, \id_{\R^m}) = J(\sigma)\,.
\end{align}
This establishes \cref{eq:fisher ieq} and hence \cref{eq:geom BL entropic} for states that are sufficiently regular.
While \cref{eq:fisher covariant} need not apply in general, the Fisher information $J(\sigma)$ and the de Bruijn identity~\eqref{eq:de bruijn} have been generalized to arbitrary states with finite second moments~\cite{dePalma18},  as discussed above. The quantity~$J(\sigma,A)$ can be defined in the same manner so that \cref{eq:step1,eq:step2} hold verbatim, see~\cite[Definition~6, Propositions~6 \& 9]{dePalma21}.%
\footnote{%
The idea is the same in \cref{foot:depalma}.
One first defines an integral quantity~$\Delta(\sigma,A)$ just like for the Fisher information, except that the multivariate Gaussian random variable now has covariance matrix $A \oplus A$~\cite[Definition~5]{dePalma21}.
If~$\sigma_X$ is a state with finite energy then~$t \mapsto \Delta(\sigma,tA)$ is again continuous, increasing, and concave~\cite[Proposition~5]{dePalma21}, and hence one can define~$J(\sigma,A) := \lim_{t\downarrow 0} \frac {\Delta(\sigma,tA)} t$~\cite[Definition~6]{dePalma21}.
Then $J(\sigma,A)$ satisfies a Stam inequality~\cite[Prop.~9]{dePalma21} that implies~\eqref{eq:step1}, and it is still a linear function of~$A$ (for nonnegative linear combinations)~\cite[Prop.~6]{dePalma21}, which is what we used in~\eqref{eq:step2}.}

The analytic form in \cref{eq:analytic BL entropic} then follows from a slight extension of \cref{thm_BLrelEnt}, or more specifically the special case discussed in \cref{cor_BL}. Namely, we need to incorporate on the entropic side the finite second moment assumption from \cref{eq:geom BL entropic}. By inspection, the variational formulae from \cref{lem_variational} applied to operators with finite second moment still hold for the respective suprema only taken over operators with finite second moment. Hence, following the proof of the BL duality in \cref{thm_BLrelEnt}, we can still go from the entropic to the analytic form when assuming that the operator in exponential form on the left hand side of the analytic form has finite second moment.
\end{proof}

While the preceding discussion restricted to the geometric case, we can also consider the general case of surjective linear map~$L_k\colon \R^m \to \R^{m_k}$, as in \cref{sec:intro}. For this, write~$L_k$ as the composition of an invertible map~$M_k\in\GL(m)$ and the projection onto the first~$m_k$ coordinates. Define a unitary operator~$U_k$ on~$\LL^2(\R^m)$ by $(U_k g)(x) := g(M_k^{-1} x) / \sqrt{|\det M_k|}$. Then, $\cE_k(\rho) := \tr_{m_k+1,\dots,m}(U_k \rho U_k^\dagger)$ defines a TPCP map that is the natural quantum version of the marginalization $g\mapsto g_k$ (same notation as in \cref{eq:introII}). We leave it for future work to determine under which conditions such quantum Brascamp-Lieb inequalities hold in general.

\emph{Note added}: In follow-up work, Eq.~\eqref{eq:geom BL entropic} from Proposition \ref{prop:geom q} has been extended to the conditional case with side information \cite[Theorem 7.3]{Ligthart20} for Gaussian states, based on~\cite{Koenig16}.
Subsequently, the latter assumption was removed by De Palma and Trevisan~\cite{dePalma21}, who further generalized \cref{prop:geom q} and also fully resolved the aforementioned question.

\subsection{Entropic uncertainty relations}\label{subsec:entropic}
In this section, we explain how the duality of \cref{thm_BLrelEnt,cor_BL} offers an elegant way to prove entropic uncertainty relations (cf.~the related work~\cite{Schwonnek18}). In order to compare our uncertainty bounds with the previous literature, we work in the current subsection with the explicit logarithm function relative to base two.

\begin{example}[Maassen-Uffink]
For $\rho_A\in\St(A)$ the \emph{Maassen-Uffink entropic uncertainty relation}~\cite{MaaUff88} for two arbitrary basis measurements,
\begin{align*}
M_{\X}(\cdot)=\sum_x\langle x|\cdot|x\rangle|x\rangle\langle x|_X \; \quad \text{and}\; \quad  M_{\Z}(\cdot)=\sum_z\langle z|\cdot|z\rangle|z\rangle\langle z|_Z \,,
\end{align*}
asserts in its strengthened form~\cite{berta_10} that
\begin{align}\label{eq:MU}
H(X) + H(Z) \geq -\log c(X,Z) + H(A)\quad\text{with $c(X,Z):=\max_{x,z} | \! \left \langle x|z \right \rangle \! |^2$.}
\end{align}
The constant $c(X,Z)$ is tight in the sense that there exist quantum states that achieve equality for certain measurement maps.
\cref{eq_BL} of \cref{cor_BL} for $n=2$, $q_1=q_2=1$, $\cE_1=M_{\X}$, and $\cE_2=M_{\Z}$ then immediately gives the equivalent analytic form
\begin{align}\label{eq_BLmaassen}
\tr \exp \left( M_{\X}^\dagger( \log\omega_1) + M_{\Z}^\dagger(\log\omega_2) \right) \leq c(X,Z) \quad\forall\omega_1,\omega_2\in\St(A)\, .
\end{align}
In other words, in order to prove \cref{eq:MU} it suffices to show \cref{eq_BLmaassen}.
Now, since the logarithm is operator concave and $M_{\X}^\dagger$ is a unital map, the operator Jensen inequality~\cite{hansen03} implies
\begin{align*}
M_{\X}^\dagger(\log X_1) \leq \log M_{\X}^\dagger(X_1)\,.
\end{align*}
Together with the monotonicity of $X \mapsto \tr \exp(X)$~\cite[Theorem~2.10]{carlen_book} and the Golden-Thompson inequality\footnote{The Golden-Thompson inequality ensures that for all Hermitian matrices $H_1$ and $H_2$ we have $\tr \exp(H_1 + H_2) \leq \tr \exp(H_1) \exp(H_2)$.} \cite{golden65,thompson65}, this establishes the analytic form of \cref{eq_BLmaassen}
\begin{align*}
\tr \exp \left( M_{\X}^\dagger( \log \omega_1) + M_{\Z}^\dagger(\log \omega_2) \right)
&\leq \tr \exp \left(  \log M_{\X}^\dagger(\omega_1) + \log M_{\Z}^\dagger(\omega_2) \right) \\
&\leq \tr M_{\X}^\dagger(\omega_1) M_{\Z}^\dagger(\omega_2)\\
&\leq c(X,Z)\,.
\end{align*}
Thus, the entropic Maassen-Uffink relation \cref{eq:MU} follows from our \cref{cor_BL}.

We note that the approach of proving entropic uncertainty relations via the Golden-Thompson inequality was pioneered by Frank \& Lieb \cite{Frank2013} and is conceptually different from the original proofs that are either based on complex interpolation theory for Schatten $p$-norms~\cite{MaaUff88} or the monotonicity of quantum relative entropy~\cite{Coles12}. We refer to~\cite{coles17} for a review on entropic uncertainty relations. As a possible extension one could choose non-trivial pre-factors $q_k\neq1$ and study the optimal uncertainty bounds in that setting as well (as done in~\cite{Schwonnek18} without the $H(A)$ term).
Another natural extension is to general quantum channels instead of measurements (as detailed in~\cite{berta16,Junge18}). The constant $c(X,Z)$ from \cref{eq:MU} is \emph{multiplicative} for tensor product measurements.
However, we might ask more generally if for given measurements the optimal lower bound in \cref{eq:MU} becomes multiplicative for tensor product measurements.
This amounts to an instance of the tensorization question from \cref{eq:tensor} and we refer to~\cite{Schwonnek18,Junge18} for a discussion.
\end{example}

An advantage of our BL analysis is that it suggests tight generalizations to multiple measurements by means of the multivariate extension of the Golden-Thompson inequality~\cite{Sutter2017}.
A basic example is as follows.

\begin{example}[Six-state~\cite{Coles11}]\label{ex:six-state}
For $\rho_A\in\St(A)$ with $\text{dim}(A)=2$ and measurement maps $M_{\X}, M_{\Y}, M_{\Z}$ of the Pauli matrices $\sigma_X,\sigma_Y,\sigma_Z$ we have
\begin{align}\label{eq_UR3}
H(X) + H(Y) + H(Z) \geq 2+H(A)\,.
\end{align}
Moreover, this relation is tight in the sense that there exist quantum states that achieve equality.
Note that applying the Maassen-Uffink relation \cref{eq:MU} for any two of of the three Pauli measurements only yields the weaker bound
\begin{align*}
H(X) + H(Y) + H(Z) \geq\frac{3}{2}+\frac{3}{2}H(A)\,.
\end{align*}
The equivalent analytic form of \cref{eq_UR3} is given by \cref{cor_BL} as
\begin{align*}
\tr \exp\left(M^\dagger_{\X}(\log \omega_1) + M^\dagger_{\Y}(\log \omega_2)+M^\dagger_{\Z}(\log \omega_3)\right) \leq \frac{1}{4}  \quad\forall\omega_1,\omega_2, \omega_3\in\St(A)\,.
\end{align*}
The same steps as in the proof of the Maassen-Uffink relation, together with Lieb's triple matrix inequality~\cite{Lieb73} then yield the upper bound\footnote{Lieb's triple matrix inequality corresponds to the three matrix Golden-Thompson inequality from~\cite{Sutter2017}.}
\begin{align*}
&\int_{0}^{\infty}  \tr M_{\X}(\omega_1) \frac{1}{M_{\Z}(\omega_3)^{-1} +t} M_{\Y}(\omega_2)   \frac{1}{M_{\Z}(\omega_3)^{-1} +t} \di t  \\
&=\sum_{x,y} \bra{x} \omega_1 \ket{x}  \bra{y} \omega_2 \ket{y}    \int_{0}^{\infty}     |\bra{x}  \frac{1}{M_{\Z}(\omega_3)^{-1} +t}  \ket{y}|^2 \di t \\
&\leq \max_{x,y}  \int_{0}^{\infty}   |\bra{x} \frac{1}{M_{\Z}(\omega_3)^{-1} +t} \ket{y}|^2  \di t\,.
\end{align*}
In the penultimate step we used that
\begin{equation}\label{eq:paulis}
\begin{aligned}
&M_{\X}(\omega) = \sum_{x \in \{x_0,x_1\}} \bra{x} \omega \ket{x} \proj{x} \quad \text{where } \left\{\ket{x_0} = \frac{1}{\sqrt{2}}(1,1)^T , \, \ket{x_1} = \frac{1}{\sqrt{2}}(1,-1)^T  \right\}\\
&M_{\Y}(\omega) = \sum_{y \in \{y_0,y_1\}} \bra{y} \omega \ket{y} \proj{y} \quad \text{where } \left\{\ket{y_0} = \frac{1}{\sqrt{2}}(1,\ci)^T , \, \ket{y_1} = \frac{1}{\sqrt{2}}(1,-\ci)^T \right\}\\
&M_{\Z}(\omega) = \sum_{z \in \{z_0,z_1\}} \bra{z} \omega \ket{z} \proj{z} \quad \text{where } \left\{\ket{z_0} = (1,0)^T , \, \ket{z_1} = (0,1)^T \right\}.
\end{aligned}
\end{equation}
As $(M_{\Z}(\omega_3)^{-1} +t)^{-1}= \sum_{z} \frac{1}{\bra{z} \omega_3 \ket{z}^{-1} + t} \proj{z}$, we get
\begin{align*}
\left|\bra{x} \frac{1}{M_{\Z}(\omega_3)^{-1} +t} \ket{y}\right|^2
= \left|\sum_{z} \frac{1}{\bra{z} \omega_3 \ket{z}^{-1} + t} \langle x|z\rangle \langle z|y\rangle\right|^2\,.
\end{align*}
Together with $ \langle x|z_0\rangle \langle z_0|y\rangle = \frac{1}{2}$ and $ \langle x|z_1\rangle \langle z_1|y\rangle = \pm \frac{\ci}{2}$ for all $x \in \{x_0,x_1\},\, y \in \{y_0,y_1 \}$ we find the upper bound
\begin{align*}
\frac{1}{4}  \int_{0}^{\infty}  \Big( (\bra{z_0} \omega_3 \ket{z_0}^{-1} + t)^{-2} + (\bra{z_1} \omega_3 \ket{z_1}^{-1} +t)^{-2}  \Big)\di t
&= \frac{1}{4}  \big( \bra{z_0} \omega_3 \ket{z_0} + \bra{z_1} \omega_3 \ket{z_1}  \big)
=\frac{1}{4}\,.
\end{align*}
This then concludes the proof of the six-state entropic uncertainty relation \cref{eq_UR3}.
\end{example}

\subsection{Minimum output entropy}\label{subsec:minoutent}
The Brascamp-Lieb duality from \cref{thm_BLrelEnt,cor_BL} is also applied usefully to general quantum channels.
Recall that the \emph{minimum output entropy} of a map $\cE \in \TPCP(A,B)$ is defined by
\begin{align}\label{eq_minOutput}
H_{\min}(\cE):= \min_{\rho \in \St(A)} H\big( \cE(\rho) \big) \, .
\end{align}
The computation of minimum output entropy is in general NP-complete~\cite{beigi07}.
Nevertheless, it is a fundamental information measure~\cite{Shor2004} that has been used, e.g., to prove super-additivity of the Holevo information~\cite{hastings09}.
\Cref{cor_BL} for $n=2$, $q_1=q_2=1$, $\cE_1=\cI$, and $\cE_2 = \cE$ gives the following result.

\begin{corollary}[Minimum output entropy]\label{cor_minOutput}
For $\cE \in \TPCP(A,B)$ and $C \in \R$, the following two statements are equivalent:
\begin{align}\label{eq_minOut}
C &\leq H\big( \cE(\rho) \big) \quad \forall  \rho \in \St(A) \, , \\
\label{eq_minOutDual}
\tr\exp(\log \omega_1 + \cE^\dagger(\log \omega_2)) &\leq\exp(-C)  \quad \forall  \omega_1 \in \St(A),\,\omega_2 \in \St(B) \, .
\end{align}
Moreover, we have
\begin{align}\label{eq_minOutProp}
H_{\min}(\cE) = - \max_{\omega \in \St(B)} \lambda_{\max}(\cE^{\dagger}( \log \omega)) \, .
\end{align}
\end{corollary}

\noindent
It is unclear if the form \cref{eq_minOutProp} could give new insights on the tensorization question of when the minimal output entropy of tensor product channels becomes additive. That is, for which~$\cE, \cF \in \TPCP(A,B)$ do we have
\begin{align}\label{eq:minoutadd}
H_{\min}(\cE\otimes\cF)\stackrel{?}{=} H_{\min}(\cE) + H_{\min}(\cF)\,.
\end{align}
We note that probabilistic counterexamples are known~\cite{hastings09}, which shows that the tensorization question \cref{eq:tensor} is in general answered in the negative.

\begin{proof}[Proof of \cref{cor_minOutput}]
We give two proofs of \cref{eq_minOutProp}, one based on the variational characterization of the relative entropy from \cref{eq_varRelEnt}, and the other based on the dual formulation from \cref{eq_minOutDual}.
Using the former approach, we see that
\begin{align*}
H_{\min}(\cE)
=\min_{\rho \in \St(A)} H\big( \cE(\rho) \big)
&=\min_{\rho \in \St(A)} -D\big(\cE(\rho)\|\id\big)\\
&=\min_{\rho \in \St(A)} - \left( \max_{\omega\in\Pos(B)} \tr \cE(\rho) \log \omega - \log \tr \omega \right) \\
&=\min_{\rho \in \St(A), \omega \in \St(B)} - \tr \rho \, \cE^{\dagger}(\log \omega) \\
&=- \max_{\rho \in \St(A), \omega \in \St(B)} \tr \rho \, \cE^{\dagger}(\log \omega) \\
&= -  \max_{\omega \in \St(B)} \lambda_{\max}\big( \cE^\dagger(\log \omega) \big) \, ,
\end{align*}
where the final step follows from the variational formula of the largest eigenvalue.

Alternatively we can verify \cref{eq_minOutProp} in the analytic picture.
To see this, we note that using the equivalence between \cref{eq_minOut} and \cref{eq_minOutDual} as well as the monotonicity of the logarithm,
\begin{align}\label{eq:H_min variational}
H_{\min}(\cE)
&=-\max_{\omega_1 \in \St(A), \,\omega_2 \in \St(B)} \log \tr \exp\bigl(\log \omega_1 + \cE^\dagger(\log \omega_2)\bigr)\,.
\end{align}
Next, note that, for any Hermitian~$H$, the Golden-Thompson inequality gives
\begin{align*}
  \max_{\omega_1 \in \St(A)} \tr \exp\bigl(\omega_1 + H\bigr)
\leq \max_{\omega_1 \in \St(A)} \tr \omega_1 \exp(H)
= \lambda_{\max}(\exp(H)) = \exp(\lambda_{\max}(H)),
\end{align*}
where the second step uses again the variational formula for the largest eigenvalue.
This inequality is in fact an equality, since the upper-bound is attained if we choose~$\omega_1$ to be a projector onto an eigenvector of~$H$ with largest eigenvalue (any such~$\omega_1$ commutes with $H$).
If we use this to evaluate \cref{eq:H_min variational}, then we obtain the desired result.
\end{proof}

\begin{example}[Qubit depolarizing channel]
The minimal output entropy of the qubit depolarizing channel
\begin{align}\label{eq:qubit depol}
\cE_p \colon X \mapsto (1-p) X + p \frac{\id_{\C^2}}{2}  \tr X\quad \text{for $p \in [0,1]$}
\end{align}
is given by $H_{\min}(\cE_p) = h\big(p/2\big)$ with $h(x):=-x\log x - (1-x) \log (1-x)$ is the \emph{binary entropy function}.
In the entropic picture, this follows as the concavity of the entropy ensures that the optimizer in \cref{eq_minOutput} can always be taken to be a pure state; the unitary covariance property of the depolarizing channel then implies that we only need to evaluate the output entropy for a single arbitrary pure state. In the analytic picture, we can use \cref{eq_minOutProp} to see that
\begin{align*}
H_{\min}(\cE_p)
= - \max_{\omega \in \St(B)} \lambda_{\max}(\cE^{\dagger}( \log \omega))
= - \max_{t \in [0,1]} \Big\{ \big(1-\frac{p}{2}\big) \log t + \frac{p}{2} \log(1-t)\Big\}
= h\Big(\frac{p}{2} \Big) \, ,
\end{align*}
where the second step follows from unitary covariance and the final step uses that $t^\star = 1-p/2$ is the optimizer.
\end{example}

\subsection{Data-processing inequality}\label{sec:data-processing}
The examples given so far employed \cref{cor_BL}, but in this section we give an example that demonstrates \cref{thm_BLrelEnt} in its full strength (with $\sigma$, $\sigma_k\neq\id$).
The \emph{data-processing inequality} (DPI) for the quantum relative entropy is a cornerstone in quantum information theory~\cite{lindblad75,uhlmann77,hermes15}.
It states that, for $\rho\in\St(A)$ and $\sigma \in \Pos(A)$, the quantum relative entropy cannot increase when applying a channel~$\cE \in \TPP(A,B)$ to both arguments, i.e.,
\begin{align*}
 D\big(\cE(\rho) \| \cE(\sigma)\big)  \leq D(\rho \| \sigma) \,.
\end{align*}
The DPI is mathematically equivalent to many other fundamental results, including the strong sub-additivity of quantum entropy~\cite{LieRus73_1,LieRus73}.
Our Brascamp-Lieb duality framework fits the DPI.
That is, \cref{thm_BLrelEnt} applied for $n=1$, $q_1=1$, $\sigma_1=\cE(\sigma)$, and $C=0$ implies the following duality.

\begin{corollary}[DPI duality]\label{cor_DPI}
For $\sigma \in \Pos(A)$ and $\cE \in \TPP(A,B)$ the following inequalities hold and are equivalent:
\begin{align}
\nonumber
D\big(\cE(\rho) \| \cE(\sigma) \big) &\leq D(\rho \| \sigma)  \quad \forall  \rho \in \St(A) \, , \\
\tr \exp\big(\log \sigma + \cE^\dagger(\log \omega) \big) &\leq \tr \exp\big(\log \omega + \log \cE(\sigma)\big) \quad \forall  \omega \in \Pos(B) \, . \label{eq_BLtoDPI}
\end{align}
\end{corollary}

As a simple example for $\tr\sigma\leq\tr\rho=1$, one can immediately see that $D(\rho \| \sigma)\geq0$ by considering the trace map $\cE(\cdot)=\tr(\cdot)$. Namely, data processing for the trace map takes the trivial analytic form $\tr\log\omega\leq0$ for quantum states $\omega\in\St(A)$.

Given that the DPI is quite powerful, we suspect that \cref{eq_BLtoDPI} may be of interest too. We note that \cref{eq_BLtoDPI} does not immediately follow from existing results and thus seems novel. For example, employing the operator concavity of the logarithm, the operator Jensen inequality, and the Golden-Thompson inequality we get
\begin{align}\label{eq_Mario}
\tr \exp\big(\log \sigma + \cE^\dagger(\log \omega) \big)\leq \tr \exp \big( \log \sigma + \log\cE^{\dagger}(\omega) \big)\leq \tr \cE^{\dagger}(\omega) \sigma=\tr \omega \cE(\sigma)\,.
\end{align}
This immediately implies Hansen's multivariate Golden-Thompson inequality~\cite[Inequality~(1)]{hansen2015}, but is in general still weaker than \cref{eq_BLtoDPI} as the Golden-Thompson inequality applied to the right-hand side of \cref{eq_BLtoDPI} likewise gives
\begin{align}
\tr \exp\big(\log \omega + \log \cE(\sigma)\big)\leq\tr \omega \cE(\sigma)\,.
\end{align}
Only when $\sigma = \id$ and $\cE$ is unital does \cref{eq_BLtoDPI} simplify to $\tr \exp( \cE^\dagger(\log \omega)) \leq \tr \omega$, reducing to \cref{eq_Mario}.%
\footnote{Alternatively, this also follows directly via Jensen's trace inequality~\cite{hansen03_2}.}

\subsection{Strong data-processing inequalities}
It is a natural to study potential strengthenings of the DPI inequality and a priori it is possible to seek for additive or multiplicative improvements.
Additive strengthenings of the DPI have recently generated interest in quantum information theory~\cite{FR14,JRSWW15,Sutter2017,Sutter_book}.
Here, we consider multiplicative improvements of the DPI, which have been called \emph{strong data-processing inequalities} in the literature.
To this end, define the \emph{contraction coefficient} of~$\cE \in \TPCP(A,B)$ at $\sigma\in \St(A)$ as
\begin{align}
\eta(\sigma,\cE):= \sup_{\St(A) \ni \rho \ne \sigma} \frac{D\big(\cE(\rho) \| \cE(\sigma) \big)}{D(\rho \|\sigma)}\,.
\end{align}
The data-processing inequality then ensures that $\eta(\sigma,\cE) \leq 1$, and we say that $\cE$ satisfies a strong data-processing inequality at~$\sigma$ if~$\eta(\sigma,\cE)<1$.
\Cref{thm_BLrelEnt} for $n=1$, $C=0$, $\sigma_1=\cE(\sigma)$, and~$q_1=\eta(\sigma,\cE)^{-1}$ implies the following equivalence.

\begin{corollary}[Strong DPI duality]\label{cor_SDPI}
For $\cE \in \TPP(A,B)$, $\sigma \in \Pos(A)$, and $\eta>0$, the following two statements are equivalent:
\begin{align}
D\big( \cE(\rho) \| \cE(\sigma) \big) &\leq \eta D( \rho \| \sigma) \quad \forall \rho \in \St(A) \, , \label{eq_SDPI} \\
\tr  \exp \bigl(\log \sigma + \cE^{\dagger}(\log \omega ) \big) &\leq \norm{\exp \Bigl(\log \omega + \frac{1}{\eta} \log \cE(\sigma) \Bigr)}_{\eta } \quad \forall  \omega \in \St(B) \, .\label{eq_BL_SDPI}
\end{align}
\end{corollary}

\noindent
Thus, to determine $\eta(\sigma,\cE)$, we aim to find the smallest constant $\eta \in [0,1]$ such that \cref{eq_SDPI} or, equivalently, \cref{eq_BL_SDPI} holds.
For unital $\cE$ and maximally mixed $\sigma=\id/d$, $d:=\dim(A)$, the duality in~\cref{cor_SDPI} simplifies to
\begin{align}\label{eq_inMid}
 \log d - H\big( \cE(\rho)\big)  &\leq  \eta  \big( \log d - H(\rho) \big) \quad \forall \rho \in \St(A) \, , \\
 \label{eq_inMidfunc}
\tr \exp\bigl( \cE^\dagger(\log \omega) \bigr) &\leq d^{\frac{\eta  -1}{\eta }} \norm{\omega}_{\eta } \quad \forall \omega \in \St(B) \, .
\end{align}
Often we are also interested in the global \emph{contraction coefficient} of~$\cE$, obtained by optimizing $\eta(\sigma, \cE) $ over all $\sigma \in \St(A)$, i.e.,
\begin{align}
\eta(\cE):= \sup_{\sigma \in \St(A)} \eta(\sigma, \cE) \, .
\end{align}

\begin{example}[Qubit depolarizing channel]
For the qubit depolarizing channel $\cE_p$ from \cref{eq:qubit depol}, which is unital, we claim that
\begin{align}\label{eq:contraction-dep}
\eta\left(\frac{\id_{\C^2}}{2},\cE_p\right)=(1-p)^2\,.
\end{align}
To prove this in the entropic picture we start by recalling that $\eta(\cE_p)=(1-p)^2$~\cite{hiai16}, which already gives $\eta(\frac{\id_{\C^2}}{2},\cE_p)\leq(1-p)^2$.
Thus, it suffices to find states $\rho \in \St(A)$ such that
\begin{align}
(1-p)^2 \leq \frac{D(\cE_p(\rho) \| \frac{\id_{\C^2}}{2})}{D(\rho \| \frac{\id_{\C^2}}{2})}
=\frac{1 - H\big( \cE_p(\rho)\big)}{1 - H(\rho)}
\end{align}
up to arbitrarily small error.
The states $\rho_{\eps} = \diag(\frac{1}{2}+\eps, \frac{1}{2}-\eps)$ satisfy this condition in the limit $\eps \to 0$.
Indeed,
\begin{align}
\lim_{\eps \to 0} \frac{1 - H\big( \cE_p(\rho_{\eps})\big)}{1 - H(\rho_{\eps})}
= \lim_{\eps \to 0} \frac{1-h\big((1-p) (1/2 + \eps) + p/2 \big)}{1-h(1/2 + \eps)}
= (1-p)^2 \, ,
\end{align}
as follows from the Taylor expansion of the binary entropy function~$h(\cdot)$.

In the analytic form of \cref{eq_inMidfunc}, the statement of \cref{eq:contraction-dep} is equivalent to the claim that $\eta=(1-p)^2$ is the smallest $\eta \in [0,1]$ such that
\begin{align}
\tr \exp \big((1-p) \log \omega + \frac{p}{2} \id_{\C^2 }\tr \log \omega \big)
\leq 2^{\frac{\eta -1}{\eta}} \norm{\omega}_{\eta } \quad \text{for all} \quad \omega \in \St(B) \, .
\end{align}
Without loss of generality we can assume that $\omega = \diag(t,1-t)$ for $t \in [0,1]$. Then, the statement above simplifies to showing that $\eta=(1-p)^2$ is the smallest $\eta \in [0,1]$ such that
\begin{align}
\big(t (1-t) \big)^{\frac{p}{2}} \big(t^{1-p} +(1-t)^{1-p} \big) \leq 2^{\frac{\eta -1}{\eta}} \big(t^\eta +(1-t)^{\eta} \big)^{\frac{1}{\eta}} \quad \text{for all} \quad t \in [0,1] \, .
\end{align}
\end{example}

\subsection{Super-additivity of relative entropy}
Another type of strengthening of the DPI is as follows.
The quantum relative entropy is \emph{super-additive} for product states in the second argument. That is, for $\rho_{AB}$, $\sigma_{AB} \in \St(A \otimes B)$ we have
\begin{align}
  D(\rho_{AB} \| \sigma_A \otimes \sigma_B) \geq D(\rho_A \| \sigma_A) + D(\rho_B \| \sigma_B) \, .
\end{align}
This directly follows from the non-negativity of the relative entropy, since $D(\rho_{AB} \| \sigma_A \otimes \sigma_B)  - D(\rho_A \| \sigma_A) - D(\rho_B \| \sigma_B) = D(\rho_{AB} \| \rho_A \otimes \rho_B) \geq 0$. If the state in the second argument is not a product state we can apply the DPI twice and find
\begin{align}
D(\rho_{AB} \| \sigma_{AB}) \geq t D(\rho_A \| \sigma_A) + (1-t) D(\rho_B \| \sigma_B) \quad \text{for all } t \in[0,1] \, .
\end{align}
A natural question is thus to find parameters $\alpha(\sigma_{AB}), \beta(\sigma_{AB})$ with $\alpha(\sigma_A \otimes \sigma_B) = \beta(\sigma_A \otimes \sigma_B) =1$ such that\footnote{We might also ask for $\alpha(\sigma_{AB}) + \beta(\sigma_{AB}) \geq 1$.}
\begin{align} \label{eq_goalEq}
D(\rho_{AB} \| \sigma_{AB}) \geq \alpha(\sigma_{AB}) D(\rho_A \| \sigma_A) + \beta(\sigma_{AB}) D(\rho_B \| \sigma_B) \, .
\end{align}
Recently, it was shown~\cite{capel18} that \cref{eq_goalEq} indeed holds for
\begin{align}\label{eq_Perez}
\alpha(\sigma_{AB}) = \beta(\sigma_{AB}) =  \Big( 1 + 2 \Big\|\sigma_A^{-\frac{1}{2}}\otimes \sigma_B^{-\frac{1}{2}} \sigma_{AB} \sigma_A^{-\frac{1}{2}}\otimes \sigma_B^{-\frac{1}{2}}- \id_{AB}\Big\|_\infty \Big)^{-1}\,.
\end{align}
Applying \cref{thm_BLrelEnt} for $n=2$, $\sigma_1=\sigma_A$, $\sigma_2=\sigma_B$, $C=0$, $\cE_1=\tr_B$, $\cE_2 = \tr_A$, $q_1= \alpha $, and $q_2 = \beta$ gives the following BL duality.

\begin{corollary}[Duality for super-additivity of relative entropy]
For $\sigma_{AB} \in \Pos(A \otimes B)$ with $\tr\sigma_{AB}=1$, $\alpha>0$, and $\beta>0$, the following two statements are equivalent:
\begin{align}
\alpha D(\rho_A \| \sigma_A) + \beta D(\rho_B \| \sigma_B) &\leq  D(\rho_{AB} \| \sigma_{AB})  \quad \forall  \rho_{AB} \in \St(A \otimes B) \label{eq_superAddIT} \,,  \\
\tr \exp( \log \sigma_{AB} + \log \omega_A + \log\omega_B )
\!&\leq\! \norm{\exp (\log \omega_A +\alpha \log \sigma_A)}_{\frac{1}{\alpha}} \norm{\exp(\log \omega_B + \beta \log \sigma_B )}_{\frac{1}{\beta}} \nonumber \\
& \hspace{40mm} \forall \omega_A \in \St(A),\omega_B \in \St(B) \, . \label{eq_superAddFunc}
\end{align}
\end{corollary}

\noindent
We leave it as an open question to find parameters $\alpha(\sigma_{AB})$ and $\beta(\sigma_{AB})$ different from \cref{eq_Perez}, satisfying \cref{eq_superAddIT} or equivalently \cref{eq_superAddFunc}.

\section{Conclusion}\label{sec:conclusion}
Our fully quantum Brascamp-Lieb dualities raise a plethora of possible extensions to study.
Taking inspiration from the commutative case~\cite{Verdu17}, this could include, e.g., Gaussian optimality questions, hypercontractivity inequalities, transportation cost inequalities, strong converses in Shannon theory, entropy power inequalities~\cite{Nair19}, or algorithmic and complexity-theoretic questions~\cite{garg2018algorithmic,burgisser2018efficient,burgisser2019towards}. For some of these applications it seems that an extension of Barthe's reverse Brascamp-Lieb duality~\cite{Barthe1998} to the non-commutative setting would be useful.

\medskip\textbf{Acknowledgements.}\phantomsection\addcontentsline{toc}{section}{Acknowledgements}
MB and DS thank the Stanford Institute for Theoretical Physics for their hospitality during the time this project was initiated.
MW would like to thank Graeme Smith and JILA for their hospitality.
MW gratefully acknowledges Misha Gromov for his hospitality at IHES and for suggesting the problem discussed in \cref{sec:Gaussian}.
We thank Eric Carlen for corresponding with us about Brascamp-Lieb inequalities for relative entropies.
We thank Ernest Tan for informing us about an error in Section~\ref{subsec:entropic} in a previous version of this manuscript.
DS acknowledges support from the Swiss National Science Foundation via the NCCR QSIT as well as project No.~200020\_165843.
MB acknowledges funding by the European Research Council (ERC Grant Agreement No.~948139).
MW acknowledges support by the NWO through Veni grant no.~680-47-459 and grant OCENW.KLEIN.267, by the Deutsche Forschungsgemeinschaft (DFG, German Research Foundation) under Germany's Excellence Strategy - EXC\ 2092\ CASA - 390781972, by the BMBF through project Quantum Methods and Benchmarks for Resource Allocation (QuBRA), and by the European Research Council~(ERC) through ERC Starting Grant 101040907-SYMOPTIC.

\bibliographystyle{arxiv_no_month}
\bibliography{bibliofile}

\end{document}